\documentclass[journal,twoside,web]{ieeecolor}

\def\BibTeX{{\rm B\kern-.05em{\sc i\kern-.025em b}\kern-.08em
		T\kern-.1667em\lower.7ex\hbox{E}\kern-.125emX}}
\usepackage{generic}

\usepackage{amsmath,	amsfonts, 		amssymb}  	
\usepackage{algorithm,	algorithmic}			
\usepackage{graphicx,	graphics}				
\usepackage{array,		multirow}				
\usepackage{bm}

\usepackage{textcomp}
\usepackage{stfloats}
\usepackage{verbatim}
\usepackage{url}
\usepackage{cite}

\usepackage[switch]{lineno}

\usepackage{hyperref}
\hypersetup{hidelinks=true}

\usepackage{savesym}
\savesymbol{AND}

\usepackage{color}

\usepackage{amsthm}
\theoremstyle{definition}
\newtheorem{assum}{Assumption}
\newtheorem{lem}{Lemma}
\newtheorem{thm}{Theorem}
\newtheorem{alg}{Algorithm}
\theoremstyle{remark}
\newtheorem{rem}{Remark}

\begin{document}
	
	\title{Output Consensus on Periodic References for Constrained Multi-agent Systems Under a Switching Network}
	
	\author{
		Shibo Han, Bonan Hou, Chong Jin Ong
		\thanks{This work has been submitted to the IEEE for possible publication. Copyright may be transferred without notice, after which this version may no longer be accessible.}
		\thanks{  A preliminary version of this article appeared in the proceedings of the 10th IFAC Conference on Networked Systems \cite{han2025output}.  }
		\thanks{Shibo Han and Chong Jin Ong are with the Department of Mechanical Engineering, 
			National University of Singapore, 117576, Singapore
			(e-mail: e0846825@u.nus.edu, mpeongcj@nus.edu.sg). }
		\thanks{Bonan Hou is with the Engineering Systems and Design, 
			Singapore University of Technologyand Design, 487372 Singapore 
			(e-mail: bonan\_hou@sutd.edu.sg).}
	}
	
	\maketitle
	
	\begin{abstract}
		This work addresses the output consensus problem of constrained heterogeneous multi-agent systems under a switching network with potential communication delays, 
		where outputs are periodic and characterized by an exosystem. 
		Since periodic references have more complex dynamics, it is more challenging to track periodic references and achieve consensus on them.	
		In this paper, a model predictive control method incorporating an artificial reference and a modified cost function is proposed to track periodic references, 
		which maintains recursive feasibility even when references switch.
		Moreover, consensus protocols are proposed to achieve consensus on periodic references in different scenarios,
		in which global information such as the set of globally admissible references and the global time index are not involved.
		Theoretical analysis proves that constrained output consensus is asymptotically achieved with the proposed algorithm
		as the references of each agent converge and agents track their references while maintaining constraint satisfaction.
		Finally, numerical examples are provided to verify the effectiveness of the proposed algorithm.
	\end{abstract}


\section{introduction}
	\IEEEPARstart{C}{onsensus} is a fundamental problem in multi-agent systems (MASs).
	For achieving consensus among heterogeneous agents, the typical problem is to achieve output consensus since each agent can have different dynamics and inputs from others.
	To achieve output consensus, one common approach is for each agent to track a local reference signal which is updated based on references of its neighbors.
	When the local references among the agents reach consensus, so will the outputs of the agents if the output tracks the local reference asymptotically, see
	\cite{wieland2011internal, ong2021consensus, zhou2022semiglobal} and references therein.

	This paper aims to achieve output consensus on periodic references for agents under state and input constraints.
	Periodic references and system constraints introduce two significant challenges.
	First, while constrained consensus on constant references has been explored in \cite{nedic2010constrained} and  \cite{lin2013constrained}, 
	the extension to periodic signal tracking requires further investigation to handle their dynamics.
	Furthermore, from a practical implementation perspective, high-load data exchange should be avoided to accommodate limited communication bandwidth.
	Second, while Model Predictive Control (MPC) is a powerful paradigm for handling constraints, its direct deployment in distributed coordination is problematic.
	Specifically, each agent exchanges and updates its reference at each step to achieve consensus, which can destroy the recursive feasibility of standard MPC schemes. 
	To address this issue, it is essential to modify the MPC scheme and introduce an artificial reference \cite{limon2008mpc} and \cite{limon2015mpc}.

	Periodic references have many applications such as formation control and path planning.
	Consequently, the consensus problem for periodic references has received considerable attention \cite{yu2022decentralized, koru2024internal,fu2024cooperative, liu2024fixed, fu2023safe}.
	However, most existing works focus on unconstrained or homogeneous MASs, 
	while the studies on constrained output consensus on periodic references for heterogeneous MASs are relatively few.
	Distributed algorithms for heterogeneous MASs proposed in \cite{wang2022distributed} and \cite{li2022leader} utilize local anti-windup controllers to handle saturation,
	in which the admissibility of references is not explicitly addressed.
	Consequently, their performance cannot be guaranteed if the reference is inadmissible.
	A distributed model predictive control method is proposed in \cite{deng2024distributed} to achieve consensus on periodic outputs.
	However, at each step, the output trajectory over an entire period must be exchanged, resulting in a heavy communication burden.
	Moreover, varying topologies and communication delays are not considered.
	Alternatively, a concise scheme combining internal model principle and reference governor is developed in \cite{ong2020governor} to ensure the consensus on Lyapunov stable outputs.
	However, compared to MPC, reference governors are typically less efficient and more conservative in handling state constraint.
	A further limitation of the aforementioned methods lies in the requirement of a global time index, which hinders their fully distributed implementation.
	In summary, the constrained output consensus on periodic references remains an open challenge.
	
	Inspired by \cite{limon2008mpc, limon2015mpc} and \cite{ong2020governor}, 
	this paper develops an algorithm for the constrained output consensus on periodic references under a switching network with potential communication delay.
	Portions of this work were presented in \cite{han2025output}. 
	Compared with the conference version, this paper provides complete theoretical proofs, generalizes the analysis to accommodate communication delays, 
	and presents extensive simulation results to demonstrate the effectiveness and the necessity of the proposed method.
	The contributions are summarized as follows. 
	First, an exosystem is utilized to characterize periodic references and the set of admissible references is determined.
	Second, an MPC controller is proposed which maintains recursive feasibility even when references switch before consensus is achieved.
	Third, projection-based consensus protocols are designed to ensure that periodic references converge to a globally admissible reference without requiring any global information.

	\textit{Notation:}
	The sets of real numbers and integers are denoted by $\mathbb{R}$ and $\mathbb{N}$, respectively.
	$\mathbb{N}^+$ and $\mathbb{N}_0^+$ denote the sets of positive and non-negative integers, respectively.
	$\mathbb{N}_a^b = \{a, a+1, \dots, b\}$ where $a \leq b$.
	$I$ and $\bm{0}$ denote the identity matrix and zero matrix with appropriate dimensions.
	$T \succ 0$ means matrix $T$ is positive definite.
	$[x]_i$ denotes the $i$th element of vector $x$.
	$\left \| x \right \|_T = \sqrt{x'Tx}$ and $\left \| x \right \| = \sqrt{x'x}$.
	$\mathcal{P}_\mathbb{W}^T(w_0) = \arg\min_{w \in \mathbb{W}} \| w-w_0 \|_T^2$ and 
	$\mathcal{P}_\mathbb{W}(w_0) = \mathcal{P}_\mathbb{W}^I(w_0)$.
	A polytope $\mathbb{Z}$ is a convex set which is expressed as $\mathbb{Z} = \{ z : H z \leq h\}$.
	$\mathrm{int} (\mathbb{Z})$ is the interior of $\mathbb{Z}$.
	For the sake of conciseness, denote $\lim_{t \rightarrow \infty} x(t) - \bar{x}(t) = \bm{0}$ by $x(t) \rightarrow \bar{x}(t)$.

\section{preliminaries and problem formulation}
	Consider a network of $M$ discrete-time systems, which are referred to as agents.
	Each agent is described by
	\begin{subequations}	\label{eqn:system}
		\begin{align}
			x_i(t+1) &= A_i x_i(t) + B_i u_i(t),  \label{eqn:systemDynamics}\\
			y_i(t)   &= C_i x_i(t),      \label{eqn:output}
		\end{align}
	\end{subequations}
	where $x_i(t) \in  \mathbb{R}^{n_i}, u_i(t) \in \mathbb{R}^{m_i}$, and $y_i(t) \in \mathbb{R}^{q}$ are state, control input, and output of the $i$th agent at time $t \in \mathbb{N}$, respectively.
	$i \in \mathbb{N}_1^M$ is the index of agents.
	$t \in \mathbb{N}$ is the time index.
	
	Meanwhile, each agent is subject to state and control constraints expressed as
	\begin{align} \label{eqn:constraints}
		\begin{bmatrix}
			x_i'(t) & u_i'(t)
		\end{bmatrix}' 			\in \mathbb{Z}_i.
	\end{align}

	The agents are connected via a time-varying directed network
	described by a weighted graph $\mathcal{G}(t) = \big( \mathcal{V},\mathcal{E}(t) \big)$
	with node set $\mathcal{V} = \{1,2,...,M \}$
	and  edge  set $\mathcal{E}(t) \subseteq \mathcal{V} \times \mathcal{V}$.
	$(i,i) \in \mathcal{E}(t), \forall t \in \mathbb{N}$ by convention.
	The adjacency matrix of $\mathcal{G}(t)$ is $\mathcal{A}(t)$ with the $(i,j)$ element being  $a_{ij}(t)$ and $\sum_{j \in \mathbb{N}_1^M} a_{ij}(t) = 1$.
	Moreover, if $a_{ij}(t)$ is non-zero, it should be greater than $\bar{a}$ for some $\bar{a} > 0$.
	The set of neighbors of node $i$ is denoted by $\mathcal{N}_i(t) = \{j\in \mathcal{V}:  (i,j) \in \mathcal{E}(t) \}$.
	
	Communication delays may exist when agent receives information from its neighbors.
	Suppose agent $j$ broadcasts its state $x_j(t)$ at each time step.
	If $(i,j) \in \mathcal{E}(t_0)$, then agent $i$ has access to $x_j \big( t_0-\tau_i^j(t_0) \big)$ at $t=t_0$,
	where $\tau_i^j(t_0)$ is the communication delay between agents $i$ and $j$ at $t=t_0$.

	Assumptions about the MAS are given below.
	
	\begin{assum} \label{asm:graph}
		$\mathcal{G}(t)$ is uniformly strongly connected in the sense that for all $t\in \mathbb{N}^+$, 
		there exists a finite $T>0$ such that the union of the graphs from $t$ to $t+T$ has a directed path between any two nodes.
	\end{assum}
	
	\begin{assum} \label{asm:delay}
		$\tau_i^i(t) = 0, \tau_i^j(t) \leq {\tau}_{\max}, \forall t \in \mathbb{N}, i,j \in \mathbb{N}_1^M$.	
	\end{assum}
	
	\begin{assum} \label{asm:AB}
		The pair $(A_i, B_i)$ is controllable.
	\end{assum}
	
	\begin{assum} \label{asm:constraint}
		$\mathbb{Z}_i$ is a polytope and contains the origin in its interior for all $i \in \mathbb{N}_1^M$.
	\end{assum}

	By Assumption \ref{asm:AB}, a stabilization gain $K_i$ exists and is chosen such that $A_i^\mathrm{c} = A_i +B_i K_i$ is Schur for each $i \in \mathbb{N}_1^M$.
	Assumption \ref{asm:AB} also guarantees that for any $\chi_1, \chi_2 \in \mathbb{R}^{n_i}$, 
	there exists a control sequence $u_i(k), t = 0,1,...,n_i-1$ which steers the state of system (\ref{eqn:systemDynamics}) from $x_i(0) = \chi_1$ to $x_i(n_i) = \chi_2$.

	Within the MAS framework, 
	the desired output of agent $i$ is generated by a distributed exosystem of the form:
	\begin{subequations}	\label{eqn:exosystem_g}
		\begin{align}
			w_i(t+1) &= \pi_i \big( w_l (t-\tau_i^l(t) ) \big), l \in \mathcal{N}_i(t),      \label{eqn:wpi}    
			\\
			y^r_i(t) &= Q_e w_i(t), 	
		\end{align}
	\end{subequations}
	where $w_i(t)$ is referred to as the reference, and $\pi_i(\cdotp)$ represents the consensus protocol.
	The tracking error is defined as:
	\begin{align}
		e_i(t) = y_i(t) - y^r_i(t)
	\end{align}	
	In the specific case where MASs have achieved consensus on the desired output, 
	the distributed exosystem (\ref{eqn:exosystem_g}) simplifies to:
	\begin{subequations}	\label{eqn:exosystem}
		\begin{align}
			w_i(t+1) &= Sw_i(t),        \label{eqn:wSw} \\
			y^r_i(t) &= Q_e w_i(t), 	\label{eqn:yr_Qw}
		\end{align}
	\end{subequations}
	Assumptions regarding the synchronized exosystem (\ref{eqn:exosystem}) are introduced below.

	\begin{assum} \label{asm:S}
		There exists $\rho \in \mathbb{N}^+$ such that $S^{\rho} = I$.	
	\end{assum}
	\begin{assum}  \label{asm:ABCS}
		For all $i \in \mathbb{N}_1^M$,
		$\begin{bmatrix}
			A_i - \lambda I & B_i \\ C_i & \bm{0}
		\end{bmatrix}$ is full row rank for each $\lambda$ which is an eigenvalue of $S$.
	\end{assum}
	
	\begin{rem}
		Assumption \ref{asm:S} ensures that $w_i(t+\rho) = S^{\rho} w_i(t) =w_i(t), \forall t \in \mathbb{N}$.
		Thus, $w_i(t)$ and $y_i^r(t)$ are periodic.
		Assumption \ref{asm:S} also implies that $S$ is Lyapunov stable and the eigenvalues of $S$ are on the unit circle.
		Assumption \ref{asm:ABCS} is a standard condition in the output regulation literature.
		In practice, this assumption is mild and generally satisfied provided that 
		$C_i$ has full row rank, $m_i \ge p_i$, 
		and the transmission zeros of system (\ref{eqn:system}) do not coincide with the frequency modes of the exosystem $S$.
	\end{rem}

	The objective of this paper is to develop distributed controllers for constrained MASs to achieve output consensus on periodic references, 
	i.e., $\lim_{t \rightarrow \infty} \big( y_i(t)-y_j(t) \big) = 0$ and $\lim_{t \rightarrow \infty} \big( y_i(t)-y_i(t+\rho) \big) = 0$.	
	MASs operates over a communication network satisfying Assumptions \ref{asm:graph} and \ref{asm:delay}, 
	and agents are modeled by (\ref{eqn:system}) satisfying Assumptions \ref{asm:AB} and  \ref{asm:constraint}.
	This objective is further formulated as: given Assumptions \ref{asm:S} and \ref{asm:ABCS}, for all $i,j \in \mathbb{N}_1^M$,
	\begin{itemize}
		\item design a tracking controller $u_i(t) = \kappa_i \big( x_i(t),w_i(t) \big)$ 
		such that $\begin{bmatrix}	x_i'(t) & u_i'(t)	\end{bmatrix}' 	\in \mathbb{Z}_i$  
		and $y_i(t) \rightarrow Q_e w_i(t)$ when $w_i(t+1) \rightarrow S w_i(t)$;
		\item design a consensus protocol $w_i(t+1) = \pi_i \big( w_l (t-\tau_i^l(t) ) \big)$, $l \in \mathcal{N}_i(t)$, 
		such that $w_i(t) \rightarrow w_j(t)$ and $w_i(t+1) \rightarrow S w_i(t)$.
	\end{itemize}

	\begin{rem}		
		It is worth noting that agents exchange and update their reference states according to the consensus protocol \eqref{eqn:wpi} rather than the linear dynamics \eqref{eqn:wSw}.
		Thus, the relationship $w_i(t+1) = Sw_i(t)$ generally does not hold prior to achieving consensus. 
		The reference is said to switch to $w_i(t+1)$ when $w_i(t+1) \neq Sw_i(t)$. 
		Consequently, the tracking controller $u_i(t) = \kappa_i \big( x_i(t), w_i(t) \big)$ is required to maintain feasibility even when references switch during the transient phase.
		Furthermore, it is only required that the tracking controller eliminate the tracking error $e_i(t)$ with respect to the desired output generated by the synchronized exosystem (\ref{eqn:exosystem}).
		This condition is sufficient to achieve output consensus, 
		provided that the consensus protocol successfully synchronizes the reference, 
		thereby reducing the distributed exosystem \eqref{eqn:exosystem_g} to the synchronized form \eqref{eqn:exosystem}.
	\end{rem}

	As shown in \cite{knobloch2012topics}, 
	a linear controller can be designed for unconstrained agents to track the desired output given by (\ref{eqn:exosystem}),
	which is given below.
	
	\begin{lem}  \label{lem:IMP}
		( \cite[Chapter 1]{knobloch2012topics}  ) 
		Consider system (\ref{eqn:system}) and exosystem (\ref{eqn:exosystem}) satisfying Assumptions \ref{asm:AB}, \ref{asm:S}, and \ref{asm:ABCS},
		the controller
		\begin{align} \label{eqn:statefeedback}
			u_i(t) = K_i x_i(t) + L_i w_i(t),
		\end{align}
		ensures 
		\begin{align}
			e_i(t)  \rightarrow \bm{0},   x_i(t) \rightarrow \Pi_i w_i(t),
		\end{align}
		where $A_i+B_i K_i$ is Schur, $L_i = \Gamma_i - K_i \Pi_i$ and
		\begin{align}	\label{eqn:PiGamma}
			A_i \Pi_i +B_i \Gamma_i = \Pi_i S, 
			C_i \Pi_i 			    = Q_e.
		\end{align}	
	\end{lem}	
	
	For constrained systems, some references are inadmissible, 
	which means the corresponding tracking errors cannot converge to $\bm{0}$ without violating constraints.
	Then, it is necessary to project  references onto the set of feasible references to obtain an admissible reference.
	As shown in \cite{lin2013constrained}, consensus can be achieved even when projection is involved in the consensus protocol, which is given below.
	
	\begin{lem} \label{lem:consensus}
		(\cite{lin2013constrained} )
		Consider MASs under a directed network $\mathcal{G}(t)$ satisfying Assumptions \ref{asm:graph} and \ref{asm:delay}.
		If $\mathcal{Z}_i$ is convex and $\bigcap_{i \in \mathbb{N}_1^M} \mathcal{Z}_i$ is non-empty,
		then the consensus protocol
		\begin{subequations}	\label{eqn:nycbe}
			\begin{align}	
				z_i(t+1)    &= \mathcal{P}_{\mathcal{Z}_i} \big( z_i^{ex}(t) \big), 			\label{eqn:zPzex} \\
				z_i^{ex}(t) &= \sum\nolimits_{j \in \mathcal{N}_i(t)} a_{ij}(t) z_j\big(t - \tau_i^j(t)\big), \label{eqn:zex}
			\end{align}
		\end{subequations}		
		ensures the consensus of $z_i(t)$, that is,
		$z_i(t) \rightarrow z_j(t), 	\forall i,j \in \mathbb{N}_1^M$.
		Meanwhile, for all $i \in \mathbb{N}_1^M$,
		$\lim_{t \rightarrow \infty} z_i(t) \in \bigcap_{j \in \mathbb{R}_1^M} \mathcal{Z}_j$.
	\end{lem}
	
	The consensus result is proved in \cite{lin2013constrained} with the weight of projection as $I$.
	Because all norms are equivalent in a finite-dimensional space\cite{horn2012matrix}, the above result still holds with (\ref{eqn:zPzex}) replaced by 
	$z_i(t+1)  = \mathcal{P}^{T_i}_{\mathcal{Z}_i} \big( z_i^{ex}(t) \big)$ where $T_i \succ 0$. 
	
\section{MPC for Periodic reference tracking}  \label{sec:MPC}
\subsection{Admissible References}
	Consider system (\ref{eqn:systemDynamics}), constraints (\ref{eqn:constraints}), synchronized exosystem (\ref{eqn:wSw}), and controller (\ref{eqn:statefeedback}).
	The augmented constrained closed-loop system is given by
	\begin{align} \label{eqn:augsystem}
		\begin{bmatrix}		x_i(t+1) \\ w_i(t+1)	\end{bmatrix} =
		\begin{bmatrix}
			A_i^\mathrm{c} & B_i L_i \\
			\bm{0}           & S
		\end{bmatrix}
		\begin{bmatrix}		x_i(t) \\ w_i(t)	\end{bmatrix}, \\
		\begin{bmatrix}
			I     & \bm{0} \\
			K_i   & L_i
		\end{bmatrix}
		\begin{bmatrix}		x_i(t) \\ w_i(t)	\end{bmatrix}	\in \mathbb{Z}_i.
	\end{align}
	
	Define the $0$-step admissible set as
	\begin{align}
		\mathcal{O}_0^i = \left\{  \begin{bmatrix}		x_i \\ w_i	\end{bmatrix} :
		\begin{bmatrix}
			I     & \bm{0} \\
			K_i   & L_i
		\end{bmatrix}
		\begin{bmatrix}		x_i \\ w_i	\end{bmatrix}	\in (1-\epsilon_i) \mathbb{Z}_i \right\},
	\end{align}
	where $\epsilon_i$ is a small positive constant used for constraint tightening.	
	The $k$-step admissible set $\mathcal{O}_k^i$ is defined recursively as
	\begin{align}
		\begin{aligned}
			\mathcal{O}_k^i =
			&\left\{
			\begin{bmatrix}		x_i \\ w_i	\end{bmatrix} :
			\begin{bmatrix}		I     & \bm{0} \\	K_i   & L_i		 \end{bmatrix}
			\begin{bmatrix}		x_i \\ w_i	\end{bmatrix} \in (1-\epsilon_i) \mathbb{Z}_i,
			\right. \\
			&~~~~~~~~~~\left.
			\begin{bmatrix}	A_i^\mathrm{c} & B_i L_i \\	\bm{0}   & S		\end{bmatrix}
			\begin{bmatrix}		x_i\\ w_i	\end{bmatrix} \in \mathcal{O}_{k-1}^i
			\right\}.
		\end{aligned}
	\end{align}
	
	$\mathcal{O}_\infty^i$ is given as $\mathcal{O}_\infty^i = \lim\nolimits_{k \rightarrow \infty} \mathcal{O}_k^i$, 
	which is the so-called maximal constraint admissible set.
	The method to determine $\mathcal{O}_\infty^i$  can be found in \cite{gilbert1991linear}. 
	Key properties of $\mathcal{O}_\infty^i$ come from its definition and are summarized as follows.
	\begin{lem}		\label{lem:Oinfty}
		For all $ 	\big[	x_i'(t) ~~ w_i'(t)	\big]' \in \mathcal{O}_\infty^i$, the following properties holds:
		1) $\big[	x_i'(t) ~~ \big( K_i x_i(t) + L_i w_i(t) \big)	\big]' \in (1-\epsilon_i) \mathbb{Z}_i $,
		2) $\big[	x_i'(t+1) ~~ w_i'(t+1)	\big]' \in \mathcal{O}_\infty^i$, 
			where $x_i(t+1)$ and $w_i(t+1)$ are determined according to (\ref{eqn:augsystem}).
	\end{lem}

	According to Lemma \ref{lem:IMP}, 
	the controller (\ref{eqn:statefeedback}) guarantees asymptotic output tracking, i.e., $e_i(t) \rightarrow \bm{0}$.
	Consequently, for any $w_i(t)$, if there exists $x_i(t)$ such that
	$\begin{bmatrix}		x_i'(t) & w_i'(t)	\end{bmatrix}' \in \mathcal{O}_\infty^i$,
	controller (\ref{eqn:statefeedback}) guarantees the elimination of $e_i(t)$ and constraint satisfaction.
	Such a reference $w_i(t)$ and the corresponding output $y_i^r(t)$ are said to be admissible.
	The set of admissible references is given as
	\begin{align}
		\mathcal{R}_\infty^i = \left\{
		w_i : \exists x_i \text{ such that } \begin{bmatrix}		x_i' & w_i'	\end{bmatrix}' \in \mathcal{O}_\infty^i
		\right\}.
	\end{align}
	
	The properties of $\mathcal{R}_\infty^i$ are summarized as follows.
	
	\begin{lem} \label{lem:Rinfty}
		For all $w_i(t)  \in \mathcal{R}_\infty^i$, the following properties hold:
		(i) $S w_i(t) \in \mathcal{R}_\infty^i$,
		(ii) $S^{-1} w_i(t) \in \mathcal{R}_\infty^i$,
		(iii) $S \mathcal{R}_\infty^i = \mathcal{R}_\infty^i$,
		(iv) $\begin{bmatrix}		\big( \Pi_i w_i(t) \big)' & w_i'(t)	\end{bmatrix}' \in \mathcal{O}_\infty^i$.
	\end{lem}
	
	\begin{proof}
		$w_i(t) \in \mathcal{R}_\infty^i$ implies that there exists $x_i(t)$ such that
		$\begin{bmatrix}		x_i'(t) ~~ w_i'(t)	\end{bmatrix}' \in \mathcal{O}_\infty^i$.
		Then, according to the property of $ \mathcal{O}_\infty^i$, we have
		$\begin{bmatrix}		x_i'(t+1) ~~ w_i'(t+1)	\end{bmatrix}' \in \mathcal{O}_\infty^i$,
		where $x_i(t+1)$ and $w_i(t+1)$ are determined according to (\ref{eqn:augsystem}).
		Thus, (i) holds.

		Since $S^\rho = I$, $S^{-1}w_i(t) = S^{\rho}S^{-1}w_i(t) = S^{\rho-1} w_i(t)$.
		According to (i), $w_i(t) \in \mathcal{R}_\infty^i$ implies $S^{\rho-1} w_i(t) \in \mathcal{R}_\infty^i$.
		Thus, $S^{-1} w_i(t) \in \mathcal{R}_\infty^i$ and (ii) holds.
		
		Property (i)  implies  $S \mathcal{R}_\infty^i \subseteq \mathcal{R}_\infty^i$.
		Property (ii) implies  $S^{-1} \mathcal{R}_\infty^i \subseteq \mathcal{R}_\infty^i$.
		Then, $\mathcal{R}_\infty^i = S^{-1}  S \mathcal{R}_\infty^i \subseteq S \mathcal{R}_\infty^i$.
		Thus, $S \mathcal{R}_\infty^i = \mathcal{R}_\infty^i$, and (iii) holds.	
		
		Consider $\begin{bmatrix}		x_i'(t) ~~ w_i'(t)	\end{bmatrix}' \in \mathcal{O}_\infty^i$.
		Lemma \ref{lem:Oinfty} implies $\lim\nolimits_{s \rightarrow \infty} \begin{bmatrix}	x_i'(t+s\rho) & w_i'(t+s\rho)	\end{bmatrix}' \in \mathcal{O}_\infty^i$.
		Lemma \ref{lem:IMP} implies $\lim_{s  \rightarrow \infty} \big( x_i(t+s\rho)-\Pi_i w_i(t+s\rho) \big) = \bm{0}$.
		Meanwhile,	$w_i (t+s\rho)  = w_i(t)$ since $S^\rho = I$.
		Thus, $\big[		\big( \Pi_i w_i(t) \big)' ~~ w_i'(t)	\big]' \in \mathcal{O}_\infty^i$ and (iv) holds.	
	\end{proof}

\subsection{MPC Controller}
	Under distributed consensus protocols, 
	the relationship $w_i(t+1) = Sw_i(t)$ typically does not hold before consensus on references is achieved.
	Inspired by \cite{limon2008mpc} and \cite{limon2015mpc},
	artificial references $\bar{w}_i(k|t)$ are introduced to guarantee recursive feasibility even when $w_i(t+1) \neq Sw_i(t)$.
	In addition, auxiliary control variables $v_i(k|t)$ are employed to provide additional degrees of freedom, 
		which is essential for avoiding constraint violations during transients and facilitating tracking.
	The proposed MPC controller is then given as
	\begin{align} \label{eqn:MPCcontroller}
		\kappa_i \big( x_i(t), w_i(t) \big) = K_i x_i(t) + L_i \bar{w}_i^*(0|t) + v^*_i(0|t),
	\end{align}
	where $\bar{w}_i^*(0|t)$ and $v^*(0|t)$ are determined by the quadratic optimization problem $\mathbb{QP}_i$ given as
	\begin{align} \label{eqn:QP}
		\min\nolimits_{\overline{\bm{W}}_i(t), \bm{V}_i(t)} ~ J_i \big( \overline{\bm{W}}_i(t), \bm{V}_i(t), \bm{X}_i(t),  w_i(t) \big)
	\end{align}
	subject to
	\begin{subequations} \label{eqn:QPconstraints}
		\begin{align}
			&~~~~~~~~~~~~
			x_i(0|t) = x_i(t),  \label{eqn:QP_initialconstraint}
			\\
			&\begin{aligned}
				&\begin{bmatrix}		x_i(k+1|t) \\ \bar{w}_i(k+1|t)	\end{bmatrix} = \\
				&~~~~~~~~~~
				\begin{bmatrix}
					A_i^\mathrm{c} & B_i L_i \\
					\bm{0}             & S
				\end{bmatrix}
				\begin{bmatrix}		x_i(k|t) \\ \bar{w}_i(k|t)	\end{bmatrix} +
				\begin{bmatrix}		v_i(k|t) \\ \bm{0}	\end{bmatrix},
			\end{aligned}			\label{eqn:QP_dynamicconstraint}
			\\
			&~~~~~~~~~~
			\begin{bmatrix}
				I     & \bm{0} \\
				K_i   & L_i
			\end{bmatrix}
			\begin{bmatrix}		x_i(k|t) \\ \bar{w}_i(k|t)	\end{bmatrix}	+
			\begin{bmatrix}		\bm{0} \\ v_i(k|t) 	\end{bmatrix}
			\in \mathbb{Z}_i,   \label{eqn:QP_feasibleconstraint}
			\\
			&~~~~~~~~~~
			\begin{bmatrix}		x_i'(N_i|t) & \bar{w}_i'(N_i|t)	\end{bmatrix}' \in 	\mathcal{O}_\infty^i,
			\label{eqn:QP_terminalconstraint}
			\\
			&~~~~~~~~~~~
			k \in \mathbb{N}_0^{N_i-1},
		\end{align}
	\end{subequations}
	where 
	\begin{align}
		\overline{\bm{W}}_i(t) &= [ {\bar{w}_i}'(0|t)    ~~ {\bar{w}_i}'(1|t)    ~~ \dots ~~ {\bar{w}_i}'(N_i|t) ]', \\
		\bm{V}_i(t) &= [ {v_i}'(0|t)    ~~ {v_i}'(1|t)    ~~ \dots ~~ {v_i}'(N_i-1|t) ]', \\
		\bm{X}_i(t) &= [ {x_i}'(0|t)    ~~ {x_i}'(1|t)    ~~ \dots ~~ {x_i}'(N_i|t) ]',
	\end{align}
	are predicted sequence of artificial references, control variables to be determined, and the corresponding predicted states, respectively.
	It should be pointed out that only ${\bar{w}_i}'(0|t)$ is an free decision variable, while other elements of $\overline{\bm{W}}_i(t)$ is determined by (\ref{eqn:QP_dynamicconstraint}).
	$N_i$ is the prediction horizon.

	The cost function is given as
	\begin{align}
			&J_i \big( \overline{\bm{W}}_i(t), \bm{V}_i(t), \bm{X}_i(t),  w_i(t) \big) = \left \| \bar{w}_i(0|t) - w_i(t) \right \|_{T_i}^2 \notag\\
			&+ 
			\sum\nolimits_{k=0}^{N_i-1} \left \| x_i(k|t) - \Pi_i \bar{w}_i(k|t) \right \|_{Q_i}^2 
			+	\sum\nolimits_{k=0}^{N_i-1} \left \| v_i(k|t) \right \|   _{R_i}^2 	\notag\\
			& +                   \left \| x_i(N_i|t) - \Pi_i \bar{w}_i(N_i|t) \right \|   _{P_i}^2,
	\end{align}
	where $Q_i$ and $R_i$ are symmetrical and positive definite matrices.	
	$P_i$ is determined by solving the discrete-time Lyapunov equation given by
	\begin{align} \label{eqn:APQ}
		{A_i^\mathrm{c}}' P_i A_i^\mathrm{c} - P_i +Q_i = \bm{0}.
	\end{align}
	Since $A^\mathrm{c}_i$ is Schur, (\ref{eqn:APQ}) admits a unique and positive definite solution $P_i$ for any positive definite matrix $Q_i$.
	
	$T_i$ is given as $T_i = \sum_{k=1}^{\rho} (S^k)'T_i^0 S^k$, where $T_i^0$ is symmetrical and positive definite.
	Since $S^{\rho} = I$, we have
	\begin{align}  \label{eqn:STST}
			S'T_iS
			&= S' \Big( \sum\nolimits_{k=1}^{\rho} (S^k)'T_i^0 S^k \Big) S= \sum\nolimits_{k=2}^{\rho+1} (S^k)'T_i^0 S^k \notag\\
			&= \sum\nolimits_{k=2}^{\rho} (S^k)'T_i^0 S^k + (S^{\rho}S)'T_i^0 S^{\rho}S \notag\\
			&= \sum\nolimits_{k=2}^{\rho} (S^k)'T_i^0 S^k + S'T_i^0 S = T_i.
	\end{align}
	
	\begin{thm}   \label{thm:MPC} 
		Consider system (\ref{eqn:system}) under constraints (\ref{eqn:constraints}), controller (\ref{eqn:MPCcontroller}), and exosystem (\ref{eqn:exosystem}) satisfying Assumptions \ref{asm:AB} to \ref{asm:ABCS}.
		If $\mathbb{QP}_i$ is feasible at $t$ and $N_i \geq n_i$,
		the following properties hold:
		
		(i)    (constraint satisfaction) constraint (\ref{eqn:constraints}) is satisfied;
		
		(ii)   (recursive feasibility) $\mathbb{QP}_i$  is feasible at $t+1$;
		
		(iii)  (asymptotic stability)  $e_i(t) \rightarrow \bm{0}$ if $w_i(t) \in  \mathcal{R}_{\infty}^i $.
	\end{thm}
	
	\begin{proof}
		See Appendix.
	\end{proof}
	
	Similar to the MPC schemes incorporated with an artificial reference proposed in \cite{limon2008mpc} and \cite{limon2015mpc}, 
	constraints (\ref{eqn:QPconstraints}) are independent of $w_i(t)$.
	Thus, the recursive feasibility of the tracking controller still holds even when $w_i(t+1) \neq S w_i(t)$.
	In contrast, the MPC scheme proposed in \cite{wang2023constrained} does not consider the case when $w_i(t)$ changes.
	This key property is summarized as follows.

	\begin{lem}
		If $\mathbb{QP}_i$ is feasible at $t$ and $N_i \geq n_i$, $\mathbb{QP}_i$  is feasible at $t+1$ even when $w_i(t+1) \neq S w_i(t)$.
	\end{lem}

	In addition, the MPC formulation is more concise than that used in \cite{deng2024distributed} since only $w_i(t)$ rather than $w_i(t+k), k \in \mathbb{Z}_0^{\rho-1}$ is involved.

\section{Consensus Protocol}	\label{sec:cp}
	In contrast to the unconstrained consensus problem, because of the existence of constraints,
	references should converge to a value which is admissible for all agents, that is,
	\begin{align}
		\lim\nolimits_{t \rightarrow \infty} w_i(t) \in \bigcap\nolimits_{j \in \mathbb{N}_1^M} \mathcal{R}_\infty^j.
	\end{align}
	$\bigcap_{j \in \mathbb{N}_1^M} \mathcal{R}_\infty^j$ is referred to as the set of globally admissible references.
	It should be noted that Assumption \ref{asm:constraint} guarantees $\bm{0} \in \mathrm{int} \mathcal{R}_\infty^j$ \cite{gilbert1991linear}.
	Thus, the set of globally admissible references is non-empty.
	In this section, consensus protocols and results are presented under three typical scenarios.
	
\subsection{Communication Without Delay}
	When there are no communication delays, the consensus protocol with diffusive terms and projection is given as
	\begin{subequations}     \label{eqn:consensus_protocol}
		\begin{align}
			w_i(t+1)    &= S  \mathcal{P}_{\mathcal{R}_\infty^i}^{T_i} \big(  w^{ex}_i(t)  \big), \\
			w^{ex}_i(t) &= \sum\nolimits_{j \in \mathcal{N}_i(t)} a_{ij}(t)  w_j(t).
		\end{align}
	\end{subequations}
	
	In contrast to the consensus protocol given in Lemma \ref{lem:consensus}, the matrix $S$ is incorporated to accommodate the desired dynamics of exosystem (\ref{eqn:wSw}).
	Meanwhile, there are more requirements on the set which $w^{ex}_i(t)$ is projected to and the projection weight.
	As shown in Section \ref{sec:MPC}, for all $i \in \mathbb{N}_1^M$,
	$S \mathcal{R}_\infty^i = \mathcal{R}_\infty^i$ and $S'T_i S = T_i$.
	These two properties lead to an important property of projection, which is given below.
	
	\begin{lem} \label{lem:projectionS}
		If $S \mathcal{Z} = \mathcal{Z}$, $T \succ 0$, and $S'TS = T$, then
		\begin{align}
			\mathcal{P}_{\mathcal{Z}}^T (Sz) = S\mathcal{P}_{\mathcal{Z}}^T (z).
		\end{align}
	\end{lem}
	
	\begin{proof}
		According to Assumption \ref{asm:S}, $S$ has full rank.
		Thus, $r \in \mathcal{Z}$ is equivalent to $Sr \in S\mathcal{Z} = \mathcal{Z}$.
		Then,
		\begin{align}
			\begin{aligned}
				\mathcal{P}_{\mathcal{Z}}^T (Sz)
				=&	\mathrm{arg} \min\nolimits_r    \{	\left \| r-Sz \right \|_T^2 : r \in \mathcal{Z} \} \\
				=&	\mathrm{arg} \min\nolimits_{Sr} \{	\left \| Sr-Sz \right \|_T^2 : Sr \in \mathcal{Z} \} \\
				=&  \mathrm{arg} \min\nolimits_{Sr} \{	\left \| r-z \right \|_{S'TS}^2 : Sr \in \mathcal{Z} \} \\
				=&	\mathrm{arg} \min\nolimits_{Sr} \{	\left \| r-z \right \|_{T}^2 : r \in \mathcal{Z} \}.		
			\end{aligned}			
		\end{align}
		Since neither cost nor constraint includes $S$, we have
		\begin{align}
			\begin{aligned}
				\mathcal{P}_{\mathcal{Z}}^T (Sz)
				=      &\mathrm{arg} \min\nolimits_{Sr} \{	\left \| r-z \right \|_{T}^2 : r \in \mathcal{Z} \} \\
				=	S  &\mathrm{arg} \min\nolimits_{r} \{	\left \| r-z \right \|_{T}^2 : r \in \mathcal{Z} \}
				=   S    \mathcal{P}_{\mathcal{Z}}^T (z)	.	
			\end{aligned}			
		\end{align}
		
		This completes the proof.
	\end{proof}
	
	Consensus protocol (\ref{eqn:consensus_protocol}) is fully distributed, 
		as it does not require knowledge of the global intersection set $\bigcap_{j=1}^M \mathcal{R}_\infty^j$ or a global time index $t$.
	Properties of consensus protocol (\ref{eqn:consensus_protocol}) are summarized below.
	
	\begin{thm} \label{thm:consensus}
		For a multi-agent system under the time-varying directed network $\mathcal{G}(t)$ satisfying Assumption \ref{asm:graph},
		consensus protocol (\ref{eqn:consensus_protocol}) ensures that for all $i,j \in \mathbb{N}_1^M$, 		
		(i)   $\lim_{t \rightarrow \infty} \big( w_i(t)-w_j(t) \big) = \bm{0}$, $\lim_{t \rightarrow \infty} \big( y_i^r(t)-y_j^r(t) \big) = \bm{0}$;
		(ii)  $\lim_{t \rightarrow \infty} w_i(t) \in \bigcap_{j \in \mathbb{N}_1^M} \mathcal{R}_\infty^j$;
		(iii) $\lim_{t \rightarrow \infty} \big( w_i(t+1)-Sw_i(t) \big) = \bm{0}$.
	\end{thm}
	
	\begin{proof}
		Given Assumption \ref{asm:S}, $0$ is not an eigenvalue of $S$. 
		Thus, $S$ has full rank and is invertible.
		Further, (\ref{eqn:consensus_protocol}) is equivalent to 
		\begin{align} \label{eqn:temp_cniuqf}
			S^{-(t+1)} w_i(t+1)    = S^{-t}  \mathcal{P}_{\mathcal{R}_\infty^i}^{T_i}   \big(  {w}^{ex}_i(t)  \big)	.		
		\end{align}
		
		Lemma \ref{lem:projectionS} implies $S^{-t}  \mathcal{P}_{\mathcal{R}_\infty^i}^{T_i}     \big(  {w}^{ex}_i(t)  \big) =   \mathcal{P}_{\mathcal{R}_\infty^i}^{T_i}  \big( S^{-t}  {w}^{ex}_i(t) \big)$.
		Denote $z_i(t) = S^{-t} w_i(t)$. 
		Then, (\ref{eqn:temp_cniuqf}) is equivalent to 
		\begin{align} \label{eqn:consensus_z}
			z_i(t+1) = \mathcal{P}_{\mathcal{R}_\infty^i} ^{T_i}
			\Big( \sum\nolimits_{j \in \mathcal{N}_i(t)} a_{ij}(t) z_j(t)  \Big).
		\end{align}

		According to Lemma \ref{lem:consensus}, (\ref{eqn:consensus_z})	guarantees that $z_i(t) \rightarrow z_j(t)$. 
		Thus, consensus protocol (\ref{eqn:consensus_protocol}) guarantees $z_i(t) \rightarrow z_j(t)$.
		Further, $\lim_{t \rightarrow \infty} \big( w_i(t)-w_j(t) \big) = \lim_{t \rightarrow \infty} S^t \big( z_i(t)-z_j(t) \big) = \bm{0}$,
		$\lim_{t \rightarrow \infty} \big( y_i^r(t)-y_j^r(t) \big) = Q_e \lim_{t \rightarrow \infty} \big( w_i(t)-w_j(t)  \big)= \bm{0}$.
		Thus, (i) holds.
		
		According to Lemma \ref{lem:consensus},
		$\lim_{t \rightarrow \infty} z_i(t) \in \bigcap_{j \in \mathbb{N}_1^M} \mathcal{R}_\infty^j$.
		According to Lemma \ref{lem:Rinfty},
		$\lim_{t \rightarrow \infty} w_i(t) = \lim_{t \rightarrow \infty} S^t z_i(t) \in \bigcap_{j \in \mathbb{N}_1^M} S^t \mathcal{R}_\infty^j = \bigcap_{j \in \mathbb{N}_1^M} \mathcal{R}_\infty^j$.
		Thus, (ii) holds.

		Property (i) leads to $w^{ex}_i(t) = w_i(t)$ and property (ii) leads to $\lim_{t \rightarrow \infty} w_i(t) \in \mathcal{R}_\infty^i$.
		Then we have
		\begin{align}	\begin{aligned}
				&\lim\nolimits_{t \rightarrow \infty} \big( w_i(t+1)-Sw_i(t) \big) \\
				= &\lim\nolimits_{t \rightarrow \infty} \big( S \mathcal{P}_{\mathcal{R}_\infty^i} ^{T_i} w_i(t)  - Sw_i(t) \big) \\
				= &\lim\nolimits_{t \rightarrow \infty} \big(Sw_i(t)-Sw_i(t) \big) = \bm{0}.
		\end{aligned}	\end{align}
		
		Thus, (iii) holds.
	\end{proof}
	
	\begin{rem}
		Although the consensus protocol (\ref{eqn:consensus_z}) and protocols proposed in \cite{deng2024distributed} and \cite{ong2020governor} guarantee consensus, 
			they are not fully distributed as they require the global time index $t$ to compute $z_i(t) = S^{-t}w_i(t)$.		
		In principle, if one were to implement (\ref{eqn:consensus_z}), 
			$\mathcal{R}_\infty^i$ could be any closed convex set containing the origin, and $T_i$ could be any positive definite matrix.
		In contrast, consensus protocol (\ref{eqn:consensus_protocol}) utilizes the specific ${\mathcal{R}_\infty^i}$ and ${T_i}$, 
			so that consensus is achieved theoretically without requiring $t$. 
		Moreover, the proposed consensus protocol is more concise and has a lower computational burden than that used in \cite{deng2024distributed} 
			since only $w_j(t)$ rather than $w_j(t+k), k \in \mathbb{Z}_0^{\rho-1}$ is involved.
	\end{rem}

\subsection{Communication With Known Delay }
	To achieve consensus on periodic references in the presence of communication delays, 
		the delay $\tau_i^j(t)$ must be explicitly compensated for in the protocol.
	When $\tau_i^j(t)$ is known, the consensus protocol is given as
	\begin{subequations}     \label{eqn:consensus_protocol_with_knownDelay}
		\begin{align} 	
			w_i(t+1)          &= S  \mathcal{P}_{\mathcal{R}_\infty^i}^{T_i}  \big(  \hat{w}^{ex}_i(t)  \big),  \\
			\hat{w}^{ex}_i(t) &= \sum\nolimits_{j \in \mathcal{N}_i(t)} a_{ij}(t) S^{\tau_i^j(t)} w_j \big( t-\tau_i^j(t) \big).
		\end{align}
	\end{subequations} 
	
	\begin{thm} \label{thm:consensus_with_knownDelay}
		For MASs under the time-varying directed network $\mathcal{G}(t)$ satisfying Assumptions \ref{asm:graph} and \ref{asm:delay},
		consensus protocol (\ref{eqn:consensus_protocol_with_knownDelay}) ensures that for all $i,j \in \mathbb{N}_1^M$, 		
		(i)   $\lim_{t \rightarrow \infty} \big( w_i(t)-w_j(t) \big) = \bm{0}$, $\lim_{t \rightarrow \infty} \big( y_i^r(t)-y_j^r(t) \big) = \bm{0}$;
		(ii)  $\lim_{t \rightarrow \infty} w_i(t) \in \bigcap_{j \in \mathbb{N}_1^M} \mathcal{R}_\infty^j$;
		(iii) $\lim_{t \rightarrow \infty} \big( w_i(t+1)-Sw_i(t) \big) = \bm{0}$.
	\end{thm}
	
	\begin{proof}
		Consensus protocol (\ref{eqn:consensus_protocol_with_knownDelay}) is equivalent to 
		\begin{align} \label{eqn:temp_cnaeiuqf}
			S^{-(t+1)} w_i(t+1)    = S^{-t}  \mathcal{P}_{\mathcal{R}_\infty^i}^{T_i}   \big(  \hat{w}^{ex}_i(t)  \big)	.		
		\end{align}
		
		Denote $z_i(t) = S^{-t} w_i(t)$. 
		Correspondingly,  $z_i \big( t-\tau_i^j(t) \big) = S^{-t+\tau_i^j(t)} w_i(t)$.
		According to Lemma \ref{lem:projectionS}, (\ref{eqn:temp_cnaeiuqf}) is equivalent to 
		\begin{align} \label{eqn:temp_cnkyuj}
			z_i(t+1) =  \mathcal{P}_{\mathcal{R}_\infty^i}^{T_i}   \big( S^{-t} \hat{w}^{ex}_i(t)  \big).
		\end{align}
		
		Meanwhile, 
		\begin{align}
			S^{-t} \hat{w}^{ex}_i(t) &= S^{-t} \sum\nolimits_{j \in \mathcal{N}_i(t)} a_{ij}(t) S^{\tau_i^j(t)} w_j \big( t-\tau_i^j(t) \big) \nonumber \\
			&=           \sum\nolimits_{j \in \mathcal{N}_i(t)} a_{ij}(t) S^{-t+\tau_i^j(t)} w_j \big( t-\tau_i^j(t) \big)     			   \nonumber \\
			&=           \sum\nolimits_{j \in \mathcal{N}_i(t)} a_{ij}(t) z_j \big( t-\tau_i^j(t) \big). 
		\end{align}
		
		Consequently, (\ref{eqn:temp_cnkyuj}) as well as (\ref{eqn:consensus_protocol_with_knownDelay}) is equivalent to 
		\begin{align} \label{eqn:temp_yjsieg}
			z_i(t+1) =  \mathcal{P}_{\mathcal{R}_\infty^i}^{T_i}  
			\Big( \sum\nolimits_{j \in \mathcal{N}_i(t)} a_{ij}(t) z_j \big( t-\tau_i^j(t) \big)  \Big).
		\end{align}

		According to Lemma \ref{lem:consensus}, (\ref{eqn:temp_yjsieg}) ensures that $z_i(t) \rightarrow z_j(t)$, and so does (\ref{eqn:consensus_protocol_with_knownDelay}).
		Then, properties (i) to (iii) can be proved following similar arguments to the proof of Theorem \ref{thm:MPC}.
	\end{proof}
	
	\begin{rem}
		Communication delay can be easily obtained if each agent $i$ broadcasts both $w(t)$ and $t$ to its neighbors.
		In this case, agent $i$ has access to $w_j \big( t-\tau_i^j(t) \big)$ and $t-\tau_i^j(t)$.
		Then, $\tau_i^j(t) = t- \big( t-\tau_i^j(t) \big)$.
	\end{rem}

\subsection{Communication With Unknown Delay}	
	When communication delays are unknown, an additional mechanism is required to estimate communication delays.
	To facilitate the estimation of $\tau_i^j(t)$, the following assumption is imposed.
	
	\begin{assum} \label{asm:delay_lower_bound}
		The lower bound of $\tau_i^j(t)$ is known, which is notated as $\underline{\tau_i^j}$.
		Meanwhile, there exists $\bar{t}$ such that for each $i,j \in \mathbb{N}_1^M$, there exist $t \leq \bar{t}$ such that $\tau_i^j(t) = \underline{\tau_i^j}$.	
	\end{assum}

	\begin{rem}
		Note that the communication delay is physically lower-bounded by a specific non-negative constant.
		Thus, $\underline{\tau_i^j}$ naturally exists.
		From a probabilistic viewpoint, as long as the probability of the event $\tau_i^j(t) = \underline{\tau}_i^j$ is non-zero,
		the probability that this event occurs at least once (i.e., there exists $t_0 \leq t$ such that $\tau_i^j(t_0) = \underline{\tau}_i^j$) increases as $t$ grows, and approaches $1$ as $t \to \infty$.
		This implies that Assumption \ref{asm:delay_lower_bound} is practically feasible for systems in continuous operation.
	\end{rem}

	In this case, agent $i$ should broadcast ${w}_i(t)$ and $\phi_i(t) \in \mathbb{N}$ to its neighbors, where $\phi_i(t+1) = \phi_i(t) + 1$.
	Meanwhile, agent $i$ receives ${w}_j \big( t - \tau_i^j(t) \big)$ and $\phi_j \big( t -\tau_i^j(t) \big)$ from its neighbor agent $j$.
	The consensus protocol is concluded as
	\begin{subequations} 	\label{eqn:consensus_protocol_with_delay}
		\begin{align} 
			w_i(t+1)              	&= S  \mathcal{P}_{\mathcal{R}_\infty^i}^{T_i}  \big(  \tilde{w}^{ex}_i(t)  \big),  \\
			\tilde{w}^{ex}_i(t)   	&= \sum\nolimits_{j \in \mathcal{N}_i(t)} a_{ij}(t) S^{\hat{\tau}_i^j(t)} {w}_j \big( t-\tau_i^j(t) \big), \\
			\hat{\tau}_i^j(t) 		&= \underline{\tau_i^j} + \Delta_i^j(t)-\min \{\Delta_i^j(\varphi) : \varphi \leq t \}, \\
			\Delta_i^j(t)         	&= \phi_i(t) - \phi_j \big( t-\tau_i^j(t) \big),     \\
			\phi_i(t+1)       		&= \phi_i(t) + 1.  \label{eqn:update_phi}
		\end{align}
	\end{subequations}
	
	\begin{thm} \label{thm:consensus_with_delay}
		For MASs under the time-varying directed network $\mathcal{G}(t)$ satisfying Assumptions \ref{asm:graph}, \ref{asm:delay}, and \ref{asm:delay_lower_bound},
		consensus protocol (\ref{eqn:consensus_protocol_with_delay}) ensures that
		for all $ i,j \in \mathbb{N}_1^M$,
		(i)   $\lim_{t \rightarrow \infty} \big( w_i(t)-w_j(t) \big) = \bm{0}$, $\lim_{t \rightarrow \infty} \big( y_i^r(t)-y_j^r(t) \big) = \bm{0}$;
		(ii)  $\lim_{t \rightarrow \infty} w_i(t) \in \bigcap_{j \in \mathbb{N}_1^M} \mathcal{R}_\infty^j$;
		(iii) $\lim_{t \rightarrow \infty} \big( w_i(t+1)-Sw_i(t) \big) = \bm{0}$.
	\end{thm}
	
	\begin{proof}
		It is derived from (\ref{eqn:update_phi}) that $\phi_i(t+1) - \phi_j (t+1) = \big( \phi_i(t) +1 \big) - \big( \phi_j (t) + 1 \big) = \phi_i(t) - \phi_j(t)$ 
			and $\phi_j \big( t-\tau_i^j(t) \big) = \phi_j ( t ) -\tau_i^j(t)$.
		Thus,  $\phi_i(t) - \phi_j ( t )$ is invariant.
		Define $\Phi_i^j = \phi_i(t) - \phi_j ( t )$.
		Then, $\Delta_i^j(t) \geq \Phi_i^j + \underline{\tau}_i^j$.
		Given Assumption \ref{asm:delay_lower_bound}, if $t \geq \bar{t}$, $\min \{\Delta_i^j(\varphi) : \varphi \leq t \} = \Phi_i^j + \underline{\tau}_i^j$.
		Further, 
		\begin{align}
			\begin{aligned}
				\hat{\tau}_i^j(t) 
				&= \underline{\tau}_i^j + \Delta_i^j(t)-\min \{\Delta_i^j(\varphi) : \varphi \leq t \} \\
				&= \underline{\tau}_i^j + \phi_i(t) - \phi_j \big( t-\tau_i^j(t) \big) -\Phi_i^j- \underline{\tau}_i^j \\
				&= \phi_i(t) - \phi_j ( t ) +\tau_i^j(t) - \Phi_i^j
				= \tau_i^j(t).
			\end{aligned}
		\end{align}
		
		Consequently, when $t \geq \bar{t}$, $\tau_i^j(t) = \hat{\tau}_i^j(t), \forall  i,j \in \mathbb{N}_1^M$.
		Then, when $t \geq \bar{t}$, consensus protocol is equivalent to (\ref{eqn:consensus_protocol_with_knownDelay}).
		Following Theorem \ref{thm:consensus_with_knownDelay},   properties (i) to (iii) hold.
	\end{proof}

\subsection{Overall Algorithm}
	Based on the above analysis, the overall algorithm is summarized in Algorithm \ref{alg:overall}.
	
	\begin{alg} \label{alg:overall}
		At each step, each agent $i$ does the following:
		\begin{enumerate}
			\item Determine control input $u_i(t) $ according to (\ref{eqn:MPCcontroller}), apply $u_i(t)$, and update $x_i(t)$.
			\item Broadcast $w_i(t)$ to its neighbors. $\phi_i(t)$ is broadcast at the same time if necessary.
			\item Receive information from its neighbors and update $w_i(t)$ according to
			(\ref{eqn:consensus_protocol}), (\ref{eqn:consensus_protocol_with_knownDelay}), or (\ref{eqn:consensus_protocol_with_delay}).
		\end{enumerate}
	\end{alg}

	With the proposed algorithm, the output consensus on periodic references for constrained MASs is achieved, which is concluded in the following theorem.
	\begin{thm} \label{thm:overall}
		Suppose Assumptions \ref{asm:graph} to \ref{asm:delay_lower_bound} hold.
		With algorithm \ref{alg:overall}, the following properties hold: $\forall i,j \in \mathbb{N}_1^M$
		(i)  $\lim_{t \rightarrow \infty} \big( w_i(t) - w_j(t) \big) = \bm{0}$,  $\lim_{t \rightarrow \infty} \big(w_i(t+1) - Sw_i(t)\big) = \bm{0}$,  
			 $w_i(t) \in \mathcal{R}_\infty^i, \forall t \in \mathbb{N}_0^+$,
		(ii) $\lim_{t \rightarrow \infty} \big( y_i(t)-Q_e w_i(t) \big) = \bm{0}$, 
		$[ x_i'(t) ~ u_i'(t) ]'   	  \in \mathbb{Z}_i, \forall t \in \mathbb{N}_0^+$,
		(iii) $\lim_{t \rightarrow \infty} \big( y_i(t)-y_j(t) \big) = \bm{0}$, and $\lim_{t \rightarrow \infty} y_i(t)-y_i(t+\rho) = \bm{0}$.
	\end{thm}
	
	\begin{proof}
		(i) Property (i) is guaranteed by Theorems \ref{thm:consensus}, \ref{thm:consensus_with_knownDelay}, or \ref{thm:consensus_with_delay}.
		
		(ii) Property (ii) is guaranteed by Theorem \ref{thm:MPC}.
		
		(iii) It is derived that 
		$y_i(t) -y_j(t) = \big( y_i(t) - Q_e w_i(t) \big) - \big( y_j(t) - Q_e w_j(t) \big) + Q_e \big( w_i(t) - w_j(t)   \big)$.
		According to (i) and (ii), we have $y_i(t) \rightarrow Q_e w_i(t)$ and $ w_i(t) \rightarrow w_j(t) $,	$\forall i,j \in \mathbb{N}_1^M$.
		Then, $y_i(t) \rightarrow y_j(t)$.
		Meanwhile, $y_i(t) \rightarrow y_i(t+\rho)$ since $w_i(t+1) \rightarrow Sw_i(t)$ and $S^\rho = I$.
		Thus, (iii) holds.
	\end{proof}

	\begin{figure}
		\centering
		\includegraphics[width=8cm]{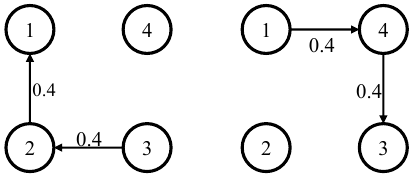}
		\caption{Communication topologies.}		\label{fig:graph}
	\end{figure}

\section{Numerical Example}
	In this section, numerical simulations involving a heterogeneous multi-agent system with four agents is presented to validate the proposed algorithm.
	Agents 1 and 2 model the linearized dynamics of a helicopter, adapted from \cite{rakovic2023model}.
	Agents 3 and 4 are modeled as discrete-time double integrators.
	The dynamics, constraints, and stabilization gains of agents are given as
	\begin{align*}
		&A_1 =
		\begin{bmatrix}
			1  &       0  &  0.4954  &  0.0026  &  -0.0069  &  -0.0596 \\
			0  &       1  &  0.0042  &  0.3896  &  -0.0688  &  -0.4395 \\
			0  &       0  &  0.9813  &  0.0083  &  -0.0454  &  -0.2459 \\
			0  &       0  &  0.0117  &  0.5813  &  -0.3898  &  -1.6662 \\
			0  &       0  &  0.0457  &  0.1274  & ~~0.8230  & ~~0.4803 \\
			0  &       0  &  0.0117  &  0.0358  & ~~0.4433  & ~~1.1361
		\end{bmatrix}, \\
		&B_1=
		\begin{bmatrix}
			~~0.0609 &   ~~0.0148 \\
			~~0.4255 &   -0.8451 \\
			~~0.2664 &   ~~0.0365 \\
			~~1.7629 &   -3.2664 \\
			-2.3152 &   ~~1.7209 \\
			-0.6083 &   ~~0.4660
		\end{bmatrix},
		C_1 =
		\begin{bmatrix}
			1  &   0 \\
			0  &   1 \\
			0  &   0 \\
			0  &   0 \\
			0  &   0 \\
			0  &   0
		\end{bmatrix}', \\
		&A_3=
		\begin{bmatrix}
			1 & 0 & 0.5 & 0   \\
			0 & 1 & 0   & 0.5 \\
			0 & 0 & 1   & 0   \\
			0 & 0 & 0   & 1
		\end{bmatrix},
		B_3 =
		\begin{bmatrix}
			0 & 0 \\
			0 & 0 \\
			1 & 0 \\
			0 & 1 
		\end{bmatrix},
		C_3 =
		\begin{bmatrix}
			1 & 0 \\
			0 & 1 \\
			0 & 0 \\
			0 & 0 
		\end{bmatrix}, 
		\\
		&\mathbb{Z}_1 = \mathbb{R}^2 \times  \mathbb{X}_1 \times \mathbb{U}_1,
		\mathbb{Z}_2 = \mathbb{R}^2 \times \mathbb{X}_2   \times \mathbb{U}_2,	
		\\
		&\mathbb{X}_1 = \{ x \in \mathbb{R}^4 : \|x\|_\infty \leq 1  \},
		\mathbb{U}_1   = \{ u \in \mathbb{R}^2 : \|u\|_\infty \leq 0.5\},	
		\\
		&\mathbb{X}_2   = \{ x \in \mathbb{R}^2 : \|x\|_\infty \leq 1 \},	
		\mathbb{U}_2   = \{ u \in \mathbb{R}^2 : \|u\|_\infty \leq 1 \},
		\\
		&K_1 = \begin{bmatrix} K^1_\mathrm{p}  & K^2_\mathrm{p} \end{bmatrix}, 
		K^1_\mathrm{p}   =
		\begin{bmatrix}
			-0.3435   &   0.1540  &   -0.9801 \\
			-0.1795   &   0.2915  &   -0.4797
		\end{bmatrix}, \\
		&K^2_\mathrm{p}  =
		\begin{bmatrix}
			0.1734  &  0.4237  & ~~0.6929 \\
			0.2935  &  0.0743  &  -0.1998
		\end{bmatrix},	\\
		&K_3 =
		\begin{bmatrix}
			-0.5 &  0    & -1  & 0 \\
			0   & -0.5 & ~~0  & -1~~ 
		\end{bmatrix}.
	\end{align*}

	The matrices $S$ and $Q_e$ of the exosystem are defined as
	\begin{align*}
		S = \mathrm{diag}(I,S_1, S_2), 
		Q_e = \begin{bmatrix} 1 & 0 & 1 & 0 & 1 & 0  \\
							  0 & 1 & 0 & 1 & 0 & 1  \end{bmatrix}.
	\end{align*}
	$S_i, i=1,2,$ are determined by 
	\begin{align*}
		S_i = e^{S_i^c T_\Delta}, 
		S_i^c = \frac{2 \pi}{\rho_i \sin (\theta_i)} \begin{bmatrix} -\cos(\theta_i) & 1 \\  -1 & \cos(\theta_i)	\end{bmatrix}, \\
		\rho_1 = 7.5, \theta_1 = 0.5 \pi, \rho_2 = 45, \theta_2 = 0.45 \pi, T_\Delta = 0.5.
	\end{align*}
	
	Consequently, the discrete-time reference signal $r(t) = \begin{bmatrix} r_1'(t) & r_2'(t)\end{bmatrix}' \in \mathbb{R}^2$ generated by the autonomous system $r(t+1) = S_i r(t)$ can be explicitly expressed as
	\begin{align*}
		r_1(t) &= \sin\big(\theta(t) \big), \\
		r_2(t) &= \sin\big(\theta(t) + \theta_i \big),\\
		\theta(t) &= \theta_0 + \frac{2 \pi}{\rho_i} T_\Delta t.
	\end{align*}
	
	In this example, the period of the reference is $\rho = 90$.
	
	$\mathcal{G}(t)$ switches randomly between two graphs shown shown in Fig. \ref{fig:graph}.
	Communication delays between agents vary randomly between $0$ to $10$, and the probability that $\tau_i^j(t) = 0$ is set to $0.1\%$.
	Consequently, $\underline{\tau_i^j} = 0, i,j \in \mathbb{N}_1^M, i \neq j$.

	The weights of MPC controller are chosen as $Q_i = I, R_i= 0.01I, T_i = \sum_{k=1}^{\rho} (S^k)'T_i^0 S^k$, where $T_i^0 = I, i \in \mathbb{N}_1^M$.
	$P_i$ is obtained by solving (\ref{eqn:APQ}). 
	The prediction horizons are chosen as $N_1 = N_2 = 6$ and $N_3 = N_4 = 4$.

	The initial conditions are given as
	\begin{align*}
		x_1 = [	0 ~~  6 ~~ 0 ~~ 0 ~~ 0 ~~  0	]',
		&x_2 = [	6 ~~  0 ~~ 0 ~~ 0 ~~ 0 ~~  0	]',\\
		x_3 = [   -6 ~~  0 ~~ 0 ~~ 0	]',
		&x_4 = [    0 ~~ -6 ~~ 0 ~~ 0	]' ,\\
		w_1 = [	1 ~~ 3 ~~ 1 ~~ 2 ~~ 1 ~~ 1	]',
		&w_2 = [	2 ~~ 2 ~~ 3 ~~ 4 ~~ 2 ~~ 2	]',\\
		w_3 = [    3 ~~ 1 ~~ 2 ~~ 3 ~~ 3 ~~ 3	]',
		&w_4 = [	4 ~~ 0 ~~ 4 ~~ 5 ~~ 4 ~~ 4	]'.
	\end{align*}
	
	Different consensus protocols are employed,
		where CP-0, CP-1, CP-2, and CP-3 represent consensus protocols  
 		(\ref{eqn:consensus_z}), (\ref{eqn:consensus_protocol}), (\ref{eqn:consensus_protocol_with_knownDelay}), and (\ref{eqn:consensus_protocol_with_delay}), respectively.
	Each protocol is tested under its intended communication conditions discussed in Section \ref{sec:cp}.
 	In addition, the following two simulations are conducted for comparison:
 	\begin{enumerate}
 		\item CP-0-$t_i$: Employ the protocol CP-0 without access to $t$. 
 		Instead, each agent uses a local time index $t_i = t + \Delta t_i$, leading to $z_i(t) = S^{-t_i} w_i(t)$, 
 		where $\Delta t_i$ is a constant offset.
 		\item CP-1-$\tau$: Neglect delay compensation and employ protocol CP-1. 
 		In this case, $w_j(t)$ in (\ref{eqn:consensus_protocol}) is replaced by $w_j \big( t-\tau_i^j(t) \big)$.
 	\end{enumerate}

	Output trajectories of the MAS and state and input trajectories of agent 1 under CP-3 are depicted in Fig. \ref{fig:Y_case2} and \ref{fig:XU}, respectively.
	Fig. \ref{fig:XU} confirms that all constraints are satisfied during the convergence to consensus.
	
	To clearly present the consensus process, trajectories of consensus error $\delta(t)$ under different consensus protocols are shown in Fig. \ref{fig:Delta}, 
		where $\delta(t)$ is defined as 
	\begin{align}
		\delta(t) = \max \{ \| y_i(t) - y_j(t) \| : i,j \in \mathbb{N}_1^M \}.
	\end{align}
	As shown in Fig. \ref{fig:Delta}, output consensus is successfully achieved by all proposed protocols under their respective nominal conditions.
	Notably, because of the equivalence between CP-0 and CP-1, the corresponding trajectories of $\delta(t)$ coincide.
	Compared to the delay-free cases, CP-2 and CP-3 exhibit slower convergence rates due to the impact of communication delays.
	Moreover, because it takes extra steps to get an accurate estimate of communication delay, 
		consensus under CP-3 is achieved after that under CP-2.
	In contrast, the consensus errors under CP-0-$t_i$ and CP-1-$\tau$ fail to converge to zero.
	This demonstrates that neglecting global clock synchronization or delay compensation leads to system failure,
		thereby highlighting the necessity of the proposed strategies.

	\begin{figure}
		\centering
		\includegraphics[width=8.8cm]{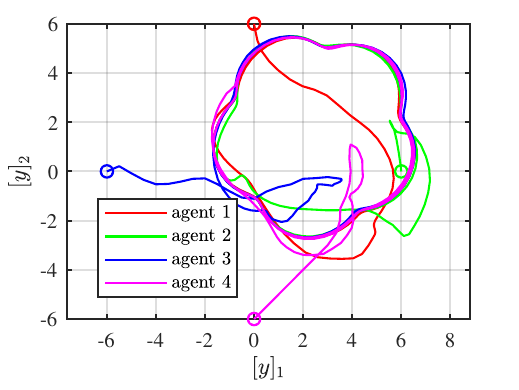}
		\caption{Output trajectories of MAS under CP-3.}		\label{fig:Y_case2}
	\end{figure}
	\begin{figure}
		\centering
		\includegraphics[width=8.8cm]{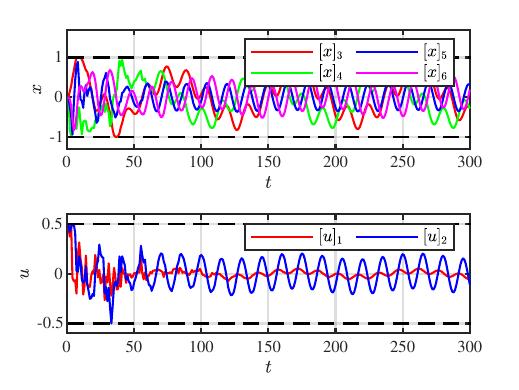}
		\caption{ State and input trajectories of agent 1 under CP-3. }		\label{fig:XU}
	\end{figure}
	\begin{figure}
		\centering
		\includegraphics[width=8.8cm]{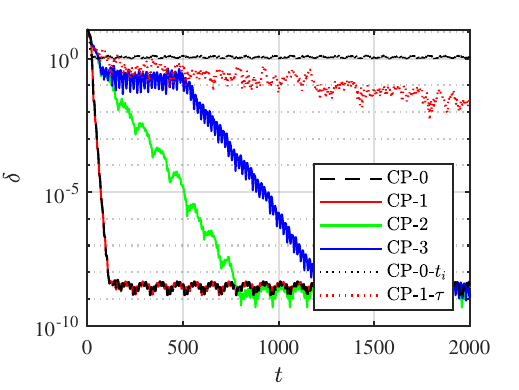}
		\caption{ Trajectory of consensus error $\delta(t)$. }		\label{fig:Delta}
	\end{figure}

	\section{Conclusion}
	This paper proposes a distributed control framework for output consensus on periodic references for constrained heterogeneous MASs.
	The periodic references are generated by a linear exosystem, yielding a consensus structure for the MPC controller and the consensus protocols.
	A novel MPC scheme incorporating an artificial reference is developed, 
	which guarantees recursive feasibility regardless of reference switching during the consensus process.
	A key theoretical contribution lies in the design of the weighting matrix for both the cost function of MPC and the projection-based consensus protocol. 
	This design guarantees the elimination of tracking error and the consensus of references without using the global time index.
	Output consensus is achieved after references reach consensus and each agent tracks the given reference.		
	In summary, the proposed algorithm works under mild assumptions without global information, making it a practical solution for the output consensus of constrained multi-agent systems.

\appendix
\section{Proof of Theorem \ref{thm:MPC}}   \label{apd:A} 
	\begin{proof}
		(i) This is ensured by (\ref{eqn:QP_feasibleconstraint}) with $k=0$.
		
		(ii) Let $\left( \overline{\bm{W}}_i^*(t), \bm{V}_i^*(t)  \right)$ denote the optimal solution to optimization problem (\ref{eqn:QP}) at $t$,
			and let $\bm{X}_i^*(t)$ denote the corresponding predicted states.
		
		At time $t+1$,  $x_i(t+1) = {x}_i^{*}(1|t)$.
		Consider the solution $\left( \overline{\bm{W}}_i^\dagger(t+1), \bm{V}_i^\dagger(t+1)  \right)$, 
			where the elements are shifted as $w_i^\dagger(k|t+1) = w_i^*(k+1|t)$ and $v_i^\dagger(k|t+1) = v_i^*(k+1|t)$, with the terminal values set to $w_i^\dagger(N_i|t+1) = S \bar{w}_i^{*'}(N_i|t)$ and $v_i^\dagger(N_i-1|t+1) = 0$.
		The corresponding predicted states can be determined and are denoted by $\bm{X}_i^\dagger(t+1)$,
			where $x_i^\dagger(k|t+1) = x_i^*(k+1|t)$ for $k \in \mathbb{N}_0^{N_i-1}$, 
			and ${x}_i^{\dagger}(N_i|t+1) = A_i^\mathrm{c}{x}_i^{*}(N_i|t) + B_iL_i  \bar{w}_i^{*}(N_i|t)$.
		Following the standard proof of MPC algorithm, it is readily to prove that 
			$\left( \overline{\bm{W}}_i^\dagger(t+1), \bm{V}_i^\dagger(t+1)  \right)$ is a feasible solution to $\mathbb{QP}_i$ at $t+1$ according to properties of $\mathcal{O}_\infty^i$.

		(iii) Denote $\bar{x}_i(k|t) = x_i(k|t) - \Pi_i \bar{w}_i(k|t), \bar{x}_i^\dagger(k|t) = x_i^\dagger(k|t) - \Pi_i \bar{w}_i^\dagger(k|t), \bar{x}_i^*(k|t) = x_i^*(k|t) - \Pi_i \bar{w}_i^*(k|t)$. 
		Then, along with (\ref{eqn:PiGamma}), it can be derived that
		\begin{align}
			\begin{aligned}
				\bar{x}_i^\dagger &(N_i|t+1) = x_i^\dagger(N_i|t+1) - \Pi_i S\bar{w}_i^*(N_i|t)  \\
				&=A_i^\mathrm{c}{x}_i^{*}(N_i|t) + B_iL_i  \bar{w}_i^{*}(N_i|t) - \Pi_i S\bar{w}_i^*(N_i|t) \\
				&=A_i^\mathrm{c} \big( x_i^*(N_i|t) - \Pi_i \bar{w}_i^*(N_i|t) \big) = A_i^\mathrm{c} \bar{x}^*_i(N_i|t).
			\end{aligned}
		\end{align}
		For the sake of conciseness, denote 
		\begin{align}
			J_i(t) 		   &= J_i \big( \overline{\bm{W}}_i(t), \bm{V}_i(t), \bm{X}_i(t),  w_i(t) \big), \\
			J_i^\dagger(t) &= J_i \big( \overline{\bm{W}}_i^\dagger(t), \bm{V}_i^\dagger(t), \bm{X}_i^\dagger(t),  w_i(t) \big), \\
			J_i^*(t)       &= J_i \big( \overline{\bm{W}}_i^*(t), \bm{V}_i^*(t), \bm{X}_i^*(t),  w_i(t) \big).
		\end{align}

		With the feasible solution given in (ii), we have
		\begin{align}
			\begin{aligned}
				&	  J^*_i(t+1) - J^*_i(t) \leq J^\dagger_i(t+1) - J^*_i(t)\\
				=&     \left \| A^\mathrm{c}_i\bar{x}_i^*(N_i|t) \right \| _{P_i}^2 + \left \| \bar{x}_i^*(N_i|t) \right \| _{Q_i}^2  - \left \| \bar{x}_i^*(N_i|t) \right \| _{P_i}^2\\
				+& \left \| S\bar{w}_i^*(t) - w_i(t+1) \right \| _{T_i}^2 - \left \| \bar{w}_i^*(t) - w_i(t) \right \| _{T_i}^2   \\
				-& \left \| \bar{x}_i^*(0|t) \right \| _{Q_i}^2 - \left \| \bar{v}_i(0|t) \right \| _{R_i}^2,
			\end{aligned}
		\end{align}
		
		According to (\ref{eqn:APQ}) and (\ref{eqn:STST}), it is derived that
		\begin{align}
			&\begin{aligned}
				&\left \| A^\mathrm{c}_i \bar{x}_i^*(N_i|t) \right \| _{P_i}^2 + \left \| \bar{x}_i^*(N_i|t) \right \| _{Q_i}^2 - \left \| \bar{x}_i^*(N_i|t) \right \| _{P_i}^2  \\
				=   &\left \|              \bar{x}_i^*(N_i|t) \right \| _{{A^\mathrm{c}_i}'P_i A^\mathrm{c}_i}^2  +  \left \| \bar{x}_i^*(N_i|t) \right \| _{Q_i}^2  -  \left \| \bar{x}_i^*(N_i|t) \right \| _{P_i}^2 \\
				=   &\left \| \bar{x}_i^*(N_i|t) \right \| _{{A^\mathrm{c}_i}'P_i A^\mathrm{c}_i + Q_i - P_i}^2  = 0,
			\end{aligned} \\
			&\begin{aligned}
				&\left \| S\bar{w}^*_i(t) -  w_i(t+1) \right \| _{T_i}^2    - \left \| \bar{w}^*_i(t) - w_i(t) \right \| _{T_i}^2 \\
				=   &\left \| S\bar{w}^*_i(t) - Sw_i(t)   \right \| _{T_i}^2    - \left \| \bar{w}^*_i(t) - w_i(t) \right \| _{T_i}^2  \\
				=   &\left \|  \bar{w}^*_i(t) -  w_i(t)   \right \| _{S'T_iS}^2 - \left \| \bar{w}^*_i(t) - w_i(t) \right \| _{T_i}^2
				=    0.
			\end{aligned}
		\end{align}
		
		Then,
		\begin{align} \label{eqn:nonincreasingJ}
			\begin{aligned}
				J^*_i(t+1) - J^*_i(t) \leq  - \left \| \bar{x}^*_i(0|t) \right \| _{Q_i}^2 - \left \| \bar{v}_i^*(0|t) \right \| _{R_i}^2.
			\end{aligned}
		\end{align}

		Since $Q_i,R_i,P_i$, and $T_i$ are all positive definite, 
		$ J_i( t )$ converges to a constant.
		This further indicates $\bar{x}^*_i(0|t) \rightarrow \bm{0}$, $ v_i^*(0|t) \rightarrow \bm{0}$, $ x_i^*(0|t) \rightarrow \Pi_i \bar{w}_i^*(0|t)$, and $J^*_i(t) = J^\dagger_i(t)$.
		Constraint (\ref{eqn:QP_terminalconstraint}) implies $\bar{w}_i^*(N_i|t) \in \mathcal{R}_\infty^i$. 
		According to Lemma \ref{lem:Rinfty}, we have $\bar{w}_i^*(0|t) \in \mathcal{R}_\infty^i$ and $\lim_{t \rightarrow \infty} \begin{bmatrix}	x_i^{*'}(0|t) & \bar{w}_i^{*'}(0|t)	\end{bmatrix}' \in 	\mathcal{O}_\infty^i$.		
		Further, it can be readily verified that $\bar{x}^*_i(k|t) \rightarrow \bm{0}$, $ x_i^*(k|t) \rightarrow \Pi_i \bar{w}_i^*(k|t)$, and $ v_i^*(k|t) \rightarrow \bm{0}$, 
			$\forall k \in \mathbb{N}_0^{N_i-1}$.
		According to the definition of $\mathcal{O}_\infty^i$, we have 
		\begin{align}	\label{eqn:nkuyd}
			\begin{bmatrix}
				I     & \bm{0} \\
				K_i   & L_i
			\end{bmatrix}
			\begin{bmatrix}		x_i^*(k|t) \\ \bar{w}^*_i(k|t)	\end{bmatrix}	+
			\begin{bmatrix}		\bm{0} \\ v_i^*(k|t) 	\end{bmatrix}
			\in \mathrm{int} \mathbb{Z}_i.
		\end{align}
		
		Then, we prove $ \bar{w}_i^*(0|t) \rightarrow w_i(t)$ by contradiction.
		Suppose $ \bar{w}_i^*(0|t) \nrightarrow w_i(t)$ and $t$ is sufficiently large.
		Consider
		\begin{align}
			\bar{w}_i(0|t) = (1-\mu) \bar{w}_i^*(0|t) + \mu w_i(t),
		\end{align}
		where $0 < \mu <1$ and $\mu$ is an infinitesimal scalar.
		Correspondingly, $\overline{\bm{W}}_i(0|t)$ is determined.
		Clearly, $\bar{w}_i(k|t) - \bar{w}^*(k|t)$ is infinitesimal for $k \in \mathbb{N}_0^{N_i}$.
		Since $\bar{w}_i^*(0|t) \in \mathcal{R}_\infty^i$ and $w_i(t) \in \mathcal{R}_\infty^i$, we have $\bar{w}_i(0|t) \in \mathcal{R}_\infty^i$.
		Then, there exists $\chi$ such that  
			$\begin{bmatrix}	\chi' & \bar{w}_i'(N_i|t)	\end{bmatrix}' \in 	\mathcal{O}_\infty^i$,
			where $\chi - x_i^{*}(N_i|t)$ is infinitesimal.
		Meanwhile, since $(A_i,B_i)$ is controllable and $N_i \geq n_i$, there exists control sequence $\bar{v}_i(k|t), k \in \mathbb{N}_0^{N_i-1}$, 
			where $\bar{v}_i(k|t) - \bar{v}_i^*(k|t)$ is infinitesimal,
			such that $x_i(N_i|t) = \chi$.		
		Since $x_i(0|t) = x_i^*(0|t)$ and $\chi - x_i^{*}(N_i|t)$ are infinitesimal, 
			$x_i(k|t) - x_i^*(k|t)$, $\bar{x}_i(k|t) - \bar{x}_i^*(k|t)$ and $v_i(k|t) - v_i^*(k|t)$ are infinitesimal, as well.
		Then, (\ref{eqn:nkuyd}) implies that
		\begin{align}	
			\begin{bmatrix}
				I     & \bm{0} \\
				K_i   & L_i
			\end{bmatrix}
			\begin{bmatrix}		x_i(k|t) \\ \bar{w}_i(k|t)	\end{bmatrix}	+
			\begin{bmatrix}		\bm{0} \\ v_i(k|t) 	\end{bmatrix}
			\in \mathbb{Z}_i.
		\end{align}
		Thus, $\left( \overline{\bm{W}}_i(t), \bm{V}_i(t)  \right)$ defined above is a feasible solution.
		
		Then, it is derived that
		\begin{align*}
			J_i(t) &- J_i^*(t) = 
				\left \| \bar{x}_i(N_i|t) \right \|   _{P_i}^2
			- \left \| \bar{x}_i^*(N_i|t) \right \|   _{P_i}^2	\\
			&+ \sum\nolimits_{k=0}^{N_i-1} \left \| \bar{x}_i(k|t) \right \|_{Q_i}^2 
			- \sum\nolimits_{k=0}^{N_i-1} \left \| \bar{x}_i^*(k|t) \right \|_{Q_i}^2 \\
			&+ \sum\nolimits_{k=0}^{N_i-1} \left \| v_i(k|t) \right \|   _{R_i}^2 
			- \sum\nolimits_{k=0}^{N_i-1} \left \| v_i^*(k|t) \right \|   _{R_i}^2	\\
			&+\| \bar{w}_i(0|t) - w_i(t)  \|_{T_i}^2- \| \bar{w}_i^*(0|t) - w_i(t)  \|_{T_i}^2.
		\end{align*}
		$ \| \bar{w}_i(0|t) - w_i(t)  \|_{T_i}^2- \| \bar{w}_i^*(0|t) - w_i(t)  \|_{T_i}^2$ is negative and infinitesimal, and the other terms are $0$ or second-order infinitesimal.
		Thus, $J_i(t) - J_i^*(t) < 0$, 	contradicting the optimality of $\left( \overline{\bm{W}}_i^*(t), \bm{V}_i^*(t)  \right)$.
		Thus, $\bar{w}_i^*(t)$ converges to $w_i(t)$.

		Further, $u_i(t)$ converges to $K_i x_i(t) + L_i w_i(t)$.
		According to Lemma \ref{lem:IMP},
		$\lim_{t \rightarrow \infty} e_i(t) = 0$.
		Thus, (iii) holds.		
	\end{proof}

\section*{References}
\bibliographystyle{IEEEtran}
\bibliography{references}

\end{document}